\documentclass[a4paper,11pt]{article}
\usepackage{graphicx} % standard LaTeX graphics tool
\usepackage{amsfonts,amsmath,amsthm,color}
\usepackage[ansinew]{inputenc}
\usepackage[T1]{fontenc}
\usepackage{tikz-cd}
\usepackage[all]{xy}
\usetikzlibrary{matrix,arrows,decorations.pathmorphing}
\newtheorem{theorem}{Theorem}[section] % 1st argument is your name for it
\newtheorem{corollary}[theorem]{Corollary}
\newtheorem{proposition}[theorem]{Proposition}
\newtheorem{definition}[theorem]{Definition}
\newtheorem{example}[theorem]{Example} 
\newtheorem{remark}[theorem]{Remark} 
\newcommand{\R}{\mathbb{R}}

\makeatletter
\renewcommand\@biblabel[1]{#1.}
\makeatother
\begin{document}
\title{Joint measurability through Naimark's theorem}
\author{Roberto Beneduci\thanks{e-mail rbeneduci@unical.it}\\
{\em Dipartimento di Fisica,}\\
{\em Universit\`a della Calabria,}\\
and\\
{\em Istituto Nazionale di Fisica Nucleare, Gruppo c. Cosenza,}\\
{\em 87036 Arcavacata di Rende (Cs), Italy,}\\
}
\date{}
\maketitle
\begin{abstract}
\noindent
 We use Naimark's dilation theorem in order to characterize the joint measurability of two POVMs. Then, we analyze the joint measurability of two commutative POVMs $F_1$ and $F_2$ which are the smearing of two self-adjoint operators $A_1$ and $A_2$ respectively. We prove that the compatibility of $F_1$ and $F_2$ is connected to the existence of two compatible self-adjoint dilations $A_1^+$ and $A_2^+$ of $A_1$ and $A_2$ respectively. As a corollary we prove that each couple of self-adjoint operators can be dilated to a couple of compatible self-adjoint operators. Next, we analyze the joint measurability of the unsharp position and momentum observables and show that it provides a master example of the scheme we propose. Finally, we give a sufficient condition for the compatibility of two effects.


\end{abstract}

\section{Introduction}
Recently there has been a renewed interest in the problem of the joint measurability (compatibility) of quantum observables in the framework of the operational approach to quantum mechanics \cite{Busch1,Reeb,Heunen,Pelloppaa,Busch2,Gudder,Heino,Spekkens,Wolf}. Such an approach rests on the use of Positive Operator Valued Measures (POVMs) in order to represent quantum observables \cite{Ali,Ali1,Prugovecki,Busch,Davies,Holevo,Ludwig} and generalizes the standard approach where a quantum observable is represented by self-adjoint operators. Indeed, self-adjoint operators are in one-to-one correspondence with Projection Valued Measures (PVMs) which define a subset of the set of POVMs. In particular, a PVM is an orthogonal POVM. 

One of the main advantage of POVMs with respect to self-adjoint operators is that two POVMs can be jointly measurable also if they do not commute while two self-adjoint operators are jointly measurable if and only if they commute. 

As a relevant physical example one can consider the case of the position and momentum observables, $Q$, $P$, in the Hilbert space $\mathcal{H}=L^2(\mathbb{R})$. Although they are incompatible they can be smeared to two jointly measurable POVMs, $F^Q$, $F^P$. Moreover, $Q$ is the sharp version of $F^Q$, i.e., $Q$ and $F^Q$ generate the same von Neumann algebra and $P$ is the sharp version of $F^P$ \cite{B4}. It is worth remarking that the existence of the compatible smearings $F^Q$ and $F^P$ is connected to the existence of two commuting dilations $Q^+$ and $P^+$ of $Q$ and $P$ in an extended Hilbert space $\mathcal{H}$ as it is illustrated by the following diagram (see example \ref{QP}).

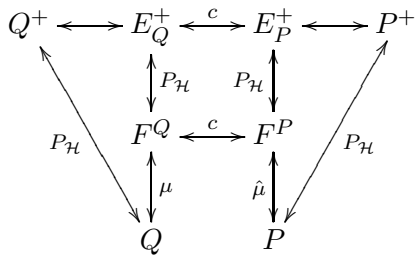
\begin{equation*}
\xymatrix{ 
Q^+\ar@{<->}[ddr]_{P_{\mathcal{H}}} \ar@{<->}[r]^{} & E_Q^+ \ar@{<->}[r]^{c} &E_P^+\ar@{<->}[r]^{} &P^+ \ar@{<->}[ddl]^{P_{\mathcal{H}}}\\ 
 & F^Q\ar@{<->}[u]_{P_{\mathcal{H}}}\ar@{<->}[r]^{c}  &  F^P\ar@{<->}[u]^{P_{\mathcal{H}}} & \\
& Q \ar@{<->}[u]_{\mu} & P \ar@{<->}[u]^{\hat{\mu}}  &
}
\end{equation*}
\noindent
where $P_{\mathcal{H}}$ is the operator of projection onto $\mathcal{H}$, $\mu$ and $\hat{\mu}$ are the Markov kernels which characterize the smearing of $Q$ and $P$, i.e., $F^Q(\Delta)=\int\mu_{\Delta}(q)\,dQ_q$, $F^P(\Delta)=\int\hat{\mu}_{\Delta}(p)\,dP_p$ and $E_Q^+$, $E_P^+$ are the Naimark's dilations of $F^Q$ and $F^P$ respectively. The symbol $\xymatrix{ \ar@{<->}[r]^{c} &}$ denotes compatibility while the symbol $\xymatrix{ \ar@{<->}[r]_{} &}$ denotes the equivalence of $Q^+$ and its spectral measure $E_Q^+$.

The aim of the present paper is to show that the scheme we just outlined for the particular case of position and momentum observables can be generalized to the case of an arbitrary couple of self-adjoint operators. In particular, we show that the joint measurability of two POVMs $F_1$, $F_2$ which are smearings of two self-adjoint operators $A_1$ and $A_2$ is connected to the existence of two commuting self-adjoint dilations $A_1^+$ and $A_2^+$ of $A_1$ and $A_2$ respectively (see theorem \ref{self}). 

The key tools in the proof of the main result are: 1) theorem \ref{CNS} where we prove that two POVMs are jointly measurable if and only if they can be dilated (Naimark's dilation) to two jointly measurable PVMs, 2) the characterization of commutative POVMs by means of Feller Markov kernels \cite{B9,B12}, 3) some previous results on the relationships between the characterization of commutative POVMs by means of Feller Markov kernels and Naimark's dilation theorem \cite{B2,B3,B6}.

As we have already said, the aim of the present work is the analysis of the joint measurability of a couple of POVMs which are the smearings of a couple of self-adjoint operators. Such a situation is very common in physics and that motivates the present work. Anyway, it is worth remarking that the extension of our results to the joint measurability of more than two POVMs is problematic. Indeed, it was recently proved  \cite{Heunen} that the characterization of the joint measurability by means of Naimark's theorem (see theorem \ref{CNS}) cannot be extended to families of more than two POVMs.

The paper is organized as follows. In section 2 we outline the main definitions and properties of POVMs, introduce the concept of Markov kernel and show that each commutative POVMs $F$ is the smearing of a self-adjoint operator $A$, i.e., $F(\Delta)=\int\mu_{\Delta}(\lambda)\,dE^A_{\lambda}=\mu_{\Delta}(A)$ where, $E^A$ is the spectral measure corresponding to $A$ and $\mu$ is a Feller Markov kernel.

In section 3, we recall the connection between the operator $A$ such that $F(\Delta)=\mu_{\Delta}(A)$ and the operator $A^+$ corresponding to the Naimark's dilation $E^+$ of $F$.  

In section 4, we prove several equivalent characterizations of the joint measurability of  two POVMs.

In section 5, we analyze the joint measurability of two POVMs which are the smearings of two self-adjoint operators and prove the main result. Then, we focus on the position and momentum observables and show that it is a master example of our scheme.

In section 6, we apply theorem \ref{CNS} to the case of two effects $E$ and $F$ and prove that they are compatible if and only if they can be dilated to two commuting projection operators $E^+$ and $F^+$ respectively. Then, we prove a sufficient condition for the joint measurability of $E$ and $F$.

\section{Definition and main properties of POVMs}
\noindent
In what follows, we denote by $\mathcal{B}(X)$ the Borel $\sigma$-algebra of a topological space $X$ and by $\mathcal{L}_s(\mathcal{H})$ the space of all bounded self-adjoint linear operators acting in a Hilbert space $\mathcal{H}$ with scalar product $\langle\cdot,\cdot\rangle$. The subspace of positive operators is denoted by $\mathcal{L}^+_s(\mathcal{H})$.

\begin{definition}
\label{POV}
A Positive Operator Valued measure (for short, POVM) is a map $F:\mathcal{B}(X)\to\mathcal{L}^+_s(\mathcal{H})$
such that:
    \begin{equation*}
    F\big(\bigcup_{n=1}^{\infty}\Delta_n\big)=\sum_{n=1}^{\infty}F(\Delta_n).
    \end{equation*}
    \noindent 
 where, $\{\Delta_n\}$ is a countable family of disjoint
    sets in $\mathcal{B}(X)$ and the series converges in the weak operator topology. It is said to be normalized if 
\begin{equation*}   
    F(X)={\bf{1}}
\end{equation*}
\noindent
where $\textbf{1}$ is the identity operator.
\end{definition}    
\begin{definition}
    A POVM is said to be commutative if
    \begin{equation}
    \big[F(\Delta_1),F(\Delta_2)\big]={\bf{0}},\,\,\,\,\forall\,\Delta_1\,,\Delta_2\in\mathcal{B}(X).
    \end{equation}
    \end{definition}

   \begin{definition}
   A POVM is said to be orthogonal if $\Delta_1\cap\Delta_2=
    \emptyset$ implies
    \begin{equation}\label{orthogonal}
    F(\Delta_1)F(\Delta_2)=\mathbf{0}
    \end{equation}
\noindent
where $\textbf{0}$ is the null operator.
\end{definition}
\begin{definition}\label{PVM}
A Spectral measure or Projection Valued measure (for short, PVM) is an orthogonal, normalized POVM.
\end{definition}
\noindent
Let $E$ be a PVM. By equation (\ref{orthogonal}), 
$$\mathbf{0}=E(\Delta)E(X-\Delta)=E(\Delta)[\mathbf{1}-E(\Delta)]=E(\Delta)-E(\Delta)^2.$$ 
\noindent
We can then restate definition \ref{PVM} as follows.
\begin{definition}
A PVM $E$ is a POVM such that $E(\Delta)$ is a projection operator for each $\Delta\in\mathcal{B}(X)$.
\end{definition}

\noindent
In quantum mechanics, non-orthogonal normalized POVMs are also called \textbf{generalised} or \textbf{unsharp} observables while PVMs are called \textbf{standard} or \textbf{sharp} observables. 

\noindent
In what follows, we shall always refer to normalized POVMs and we shall use the term ``measurable'' for the Borel measurable functions.
For any vector $\psi\in\mathcal{H}$, the map
$$\langle F(\cdot)\psi,\psi\rangle \,:\,\mathcal{B}(X)\to [0,1],
\qquad
\Delta \mapsto \langle F(\Delta)\psi,\psi\rangle,$$
is a measure. 
In the following, we shall use the symbol $d\langle F_x\psi,\psi\rangle$ to mean integration with respect to $\langle F(\cdot)\psi,\psi\rangle$.
\noindent
A measurable function $f:N\subset X\to f(N)\subset\mathbb{R}$ is said to be almost everywhere (a.e.) one-to-one with respect to a POVM $F$ if it is one-to-one on a subset $N'\subset N$ such that $F(N-N')=\mathbf{0}$.
A function $f:X\to\mathbb{R}$ is bounded with respect to a POVM $F$, if it is equal to a bounded function $g$ a.e. with respect to $F$, that is, if $f=g$ a.e. with respect to the measure $\langle F(\cdot)\psi,\psi\rangle$,  $\forall \psi \in \mathcal{H}$.
\noindent
For any real, bounded and measurable function $f$ and for any POVM $F$, there is a unique \cite{Berberian} bounded self-adjoint operator $B\in\mathcal{L}_s(\mathcal{H})$ such that
\begin{equation}
\label{6}
\langle B\psi,\psi\rangle=\int f(x)d\langle F_{x}\psi,\psi\rangle,\quad\text{for each}\quad \psi\in\mathcal{H}.
\end{equation}
If equation (\ref{6}) is satisfied, we write $B=\int f(x)dF_{x}$ or $B=\int f(x)F(dx)$ equivalently. 
\noindent
\begin{definition}
The spectrum $\sigma(F)$ of a POVM $F$ is the closed set 
$$\{x\in{X}:\,F(\Delta)\neq\mathbf{0},\,\forall\Delta\,\text{open},\,x\in\Delta\}.$$
\end{definition}
\noindent

\noindent
By the spectral theorem \cite{Reed}, there is a one-to-one correspondence between  PVMs $E$ with spectrum in $\mathbb{R}$ and self-adjoint operators $B$, the correspondence being given by
$$B=\int\lambda dE^B_{\lambda}.$$
\noindent
Notice that the spectrum of $E^B$ coincides with the spectrum of the corresponding self-adjoint operator $B$. Moreover, in this case a functional calculus can be developed. Indeed, if $f:{\mathbb R}\to{\mathbb R}$ is a measurable real-valued function, we can define the self-adjoint operator \cite{Reed}
$$f(B)=\int f(\lambda) dE^B_{\lambda}.$$
\noindent
 If $f$ is bounded, then $f(B)$ is bounded \cite{Reed}.
\noindent

The following theorem gives a characterization of commutative POVMs as smearing of spectral measures with the smearing realized by means of Feller Markov kernels.  
\begin{definition}
Let $\Lambda$ be a topological space. A Markov kernel is a map $\mu:\Lambda\times\mathcal{B}(X)\to[0,1]$ such that,
\begin{itemize}
\item[1.] $\mu_{\Delta}(\cdot)$ is a measurable function for each $\Delta\in\mathcal{B}(X)$,
\item[2.] $\mu_{(\cdot)}(\lambda)$ is a probability measure for each $\lambda\in \Lambda$.
\end{itemize}
\end{definition}
 \begin{definition}
A Feller Markov kernel is a  Markov kernel $\mu_{(\cdot)}(\cdot):\Lambda\times\mathcal{B}(X)\to[0,1]$ such that the function 
$$G(\lambda)=\int_{X}f(x)\,d\mu_x(\lambda),\quad\lambda\in\Lambda$$
\noindent
is continuous and bounded whenever $f$ is continuous and bounded. 
\end{definition}
\noindent
In the following the symbol $\mathcal{A}^W(F)$ denotes the von Neumann algebra generated by the POVM $F$, i.e., the von Neumann algebra generated by the set $\{F(\Delta)\}_{\mathcal{B}(X)}$. Hereafter, we assume that $X$ is a Hausdorff, locally compact, second countable topological space.
\begin{theorem}[\cite{B9,B1}]\label{Markov}
A POVM $F:\mathcal{B}(X)\to\mathcal{L}^+_s(\mathcal{H})$ is commutative if and only if there exists a bounded self-adjoint operator $A=\int \lambda\,dE_\lambda$ with spectrum $\sigma(A)\subset[0,1]$, a subset $\Gamma\subset\sigma(A)$, $E(\Gamma)=\mathbf{1}$, a ring $\mathcal{R}$ which generates $\mathcal{B}(X)$ and a Feller Markov Kernel  $\mu:\Gamma\times\mathcal{B}(X)\to[0,1]$ such that 
\begin{enumerate}
\item[1)] $F(\Delta)=\int_{\Gamma}\mu_{\Delta}(\lambda)\,dE_{\lambda},\quad\Delta\in\mathcal{B}(X)$.
\item[2)] $\mathcal{A}^W(F)=\mathcal{A}^W(A)$.
\item[3)] $\mu$ separates the points in $\Gamma$.
\item[4)] $\mu_{\Delta}$ is continuous for each $\Delta\in\mathcal{R}$.
\end{enumerate}
\end{theorem}

\noindent
Item 1) in theorem \ref{Markov} expresses $F$ as a smearing of $E$. Item 3) means that, for any couple of points $\lambda_1,\lambda_2\in\Gamma$, there is a $\Delta\in\mathcal{B}(X)$ such that $\mu_{\Delta}(\lambda_1)\neq\mu_{\Delta}(\lambda_2)$.

\begin{definition}
The operator $A$ introduced in theorem \ref{Markov} is called the sharp version of $F$.  
\end{definition}
\begin{theorem}\cite{B0,B1,B4}\label{bijections}
 The sharp version $A$ is unique up to almost everywhere bijections.
\end{theorem}

\begin{remark}\label{information}
Theorem \ref{Markov} defines a relationship between a commutative POVM $F$ and its sharp version $A$ which can be formalized by the introduction of an equivalence relation between $A$ and $F$ (see Ref. \cite{B5}). 
\end{remark}

\section{Sharp version as projection of a Naimark's operator}
\noindent 
In the present section, we use Naimark's dilation theorem in order to characterize the sharp version of a commutative POVM. First, we recall the Naimark's dilation theorem.  

\begin{theorem}[Naimark \cite{Naimark,Riesz,Akhiezer,Schroeck1, Paulsen}]
\label{Neumark}
Let $F$ be a POVM. Then, there exist an extended Hilbert space $\mathcal{H}^+$ and a PVM $E^+$ on $\mathcal{H}^+$ such that
$$F(\Delta)\psi=PE^+(\Delta)\psi,\quad \forall \psi\in\mathcal{H}$$
where $P$ is the operator of projection onto $\mathcal{H}$.
\end{theorem}
\noindent
Notice that Naimark's theorem assures that to each POVM $F$ acting on $\mathcal{H}$ there corresponds a PVM $E^+$ acting on an extended Hilbert space $\mathcal{H}^+$ while, theorem \ref{Markov} assures that to each commutative POVM $F$ there corresponds a PVM $E$ (the sharp version of $F$) acting on $\mathcal{H}$. The following theorem establishes a relationship between $E^+$ and $E$ in the case of real POVMs.

\begin{definition}\label{one-to-one}
Let $A$ and $B$ be self-adjoint operators. Whenever there exists a one-to-one measurable function $f$ such that $A=f(B)$, we say that $A$ is equivalent to $B$ and write $A\leftrightarrow B$. 
\end{definition}

\begin{theorem}\cite{B2,B3,B5}\label{Pro}
Let $f$ be a bounded measurable real valued function. Let $A^+$ be the self-adjoint operator corresponding to a Naimark's dilation $E^+$ of $F$. Then,
$$F(f):=\int f(t)\,d F_t=Pf(A^+)P.$$
\end{theorem}

In the case $f$ is unbounded, the domain of definition of the operators must be taken into account \cite{Lahti}.

\begin{theorem}[\cite{B6}]\label{Be2} 
Let $F:\mathcal{B}(\mathbb{R})\to\mathcal{L}^+_s(\mathcal{H})$ be a commutative POVM such that the operators in the range of $F$ are discrete\footnote{$F(\Delta)$ is discrete if it has a complete set of eigenvectors.}. Let $A$ be the sharp reconstruction of $F$ and $A^+=\int \lambda\, dE^+_{\lambda}$ the Naimark's operator corresponding to the Naimark's dilation $E^+$. Then, there are two bounded, one-to-one functions $f$ and $h$ such that
$$h(A)=\int f(t) \,d F_t=P^+f(A^+)P.$$ 
\end{theorem}
\noindent
Theorem \ref{Be2} establishes that $h(A)$ is the projection of $f(A^+)$ with $h$ and $f$ one-to-one. According to definition \ref{one-to-one} we can say that a correspondence between $A$ and $A^+$ is established as well. We denote such a correspondence by $A\leftrightarrow Pr A^+$.

\section{Conditions for the joint measurability}

In the present section we recall the definition and some of the main theorems on the joint measurability of two POVMs.  Then, we use Naimark's dilation theorem in order to prove several equivalent necessary and sufficient conditions for the joint measurability.

\begin{definition}
Two POVMs $F_1:\mathcal{B}(X_1)\to\mathcal{L}^+_s(\mathcal{H})$, $F_2:\mathcal{B}(X_2)\to\mathcal{L}^+_s(\mathcal{H})$ are compatible (or jointly measurable) if they are the marginals of a joint POVM $F:\mathcal{B}(X_1\times X_2)\to\mathcal{L}^+_s(\mathcal{H})$.  
\end{definition}

\noindent
We recall that the symbol $\mathcal{B}(X_1\times X_2)$ denotes the product $\sigma$-algebra generated by the family of sets $\{\Delta_1\times\Delta_2\,\,:\,\,\Delta_1\in\mathcal{B}(X_1),\,\Delta_2\in\mathcal{B}(X_2)\}$.

Two POVMs $F_1$ and $F_2$ commute if $[F_1(\Delta_1),F_2(\Delta_2)]=\mathbf{0}$, for each $\Delta_1\in\mathcal{B}(X_1)$ and $\Delta_1\in\mathcal{B}(X_2)$. In the following, the commutativity of two POVMs $F_1$ and $F_1$ is denoted by the symbol $[F_1,F_2]=\mathbf{0}$.

If $E_1$ and $E_2$ are two PVMs, we have the following characterizations of the compatibility.
\begin{theorem}[\cite{Lahti1}]\label{PVM}
Let $E_1$ and $E_2$ be two PVMs. The following conditions are equivalent:
\begin{enumerate}
 \item[i)] they are compatible,
\item[ii)] they are the marginals of a joint PVM $E$, 
\item[iii)] they commute.
\end{enumerate}
\end{theorem}

\noindent
Thanks to the spectral theorem which assures a one-to-one correspondence between self-adjoint operators and real PVMs (i.e., PVMs with spectrum in the reals) we can define the compatibility (joint measurability) of two self-adjoint operators. In particular, we say that $A_1$ and $A_2$ are compatible if the corresponding PVMs are compatible.  Therefore, as a consequence of the previous characterization of the compatibility of two PVMs, we have the following characterization of the compatibility of two self-adjoint operators.

\begin{corollary}
Two self-adjoint operators are compatible or jointly measurable if and only if they commute.
\end{corollary}

\noindent
As the following theorem shows, in the case of two POVMs, commutativity implies compatibility but the converse is not true, i.e, commutativity is not a necessary condition for the compatibility. That is one of the main advantage in using POVMs in order to represent quantum observables and is illustrated in example \ref{QP}.

\begin{theorem}[\cite{Lahti1}]\label{commuting}
Two commuting POVMs are compatible.
\end{theorem}

\noindent
Now, we use Naimark's dilation theorem in order to characterize the compatibility of two POVMs.

\begin{theorem}\label{CNS}
Two POVMs $F_1:\mathcal{B}(X_1)\to\mathcal{L}^+_s(\mathcal{H})$ and $F_2:\mathcal{B}(X_2)\to\mathcal{L}^+_s(\mathcal{H})$ are compatible if and only if there are two Naimark extensions $E^+_1:\mathcal{B}(X_1)\to\mathcal{L}^+_s(\mathcal{H})$ and $E^+_2:\mathcal{B}(X_2)\to\mathcal{L}^+_s(\mathcal{H})$ such that $[E^+_1,E^+_2]=\mathbf{0}$. 
\end{theorem}
\begin{proof}
Suppose $F_1$ and $F_2$ are compatible. Then, there is a POVM $F$ of which $F_1$ and $F_2$ are the marginals; i.e., $F_1(\Delta_1)=F(\Delta_1\times X_2)$, $F_2(\Delta_2)=F(X_1\times\Delta_2)$. Let $E^+$ be a Naimark dilation of $F$ and consider the PVMs $E_1^+(\Delta_1)=E^+(\Delta_1\times X_2)$ and $E^+_2(\Delta_2)=E^+(X_1\times\Delta_2)$. We have, $PE^+_1(\Delta_1)P=PE^+(\Delta_1\times X_2)P=F(\Delta_1\times X_2)=F_1(\Delta_1)$ and $PE^+_2(\Delta_2)P=PE^+(X_1\times\Delta_2)P=F(X_1\times\Delta_2)=F_2(\Delta_2)$. Moreover, $E^+_1$ and $E^+_2$ commutes since they are the marginals of the PVM $E^+$.

Conversely, suppose there are two Naimark dilation $E^+_1$ and $E^+_2$ such that $[E^+_1,E^+_2]=\mathbf{0}$. Thanks to the commutativity $[E^+_1,E^+_2]=\mathbf{0}$, there is a joint PVM $E^+$; i.e., $E^+_1(\Delta_1)=E^+(\Delta_1\times X_2)$, $E^+_2(\Delta_2)=E^+(X_1\times\Delta_2)$. We have, 
\begin{align*}
F_1(\Delta_1)&=PE^+_1(\Delta_1)P=PE^+(\Delta_1\times X_2)P=F(\Delta_1\times X_2)\\
F_2(\Delta_2)&=PE^+_2(\Delta_2)P=PE^+(X_1\times\Delta_2)P=F(X_1\times\Delta_2)
\end{align*}
\noindent
 where, $F:=PE^+P$. Therefore, $F$ is a joint POVM for $F_1$ and $F_2$.
\end{proof}
\noindent
Note that if $F_1$ and $F_2$ are PVMs, theorem \ref{CNS} coincides with theorem \ref{PVM}, iii), i.e.,  $F_1$ and $F_2$  are compatible if and only if  they commute. Indeed, $PE^+_iP=E_i$ implies that $[P,E^+_i]=\mathbf{0}$ and then $[E_1,E_2]=\mathbf{0}$.

Recently it was proved that the compatibility of more than two POVMs cannot be characterized by means of the Naimark's dilation \cite{Heunen}.

In the case of real POVMs, theorem \ref{CNS} can be expressed in the language of self-adjoint operators.

\begin{theorem}\label{CNS1}
Two real POVMs $F_1$ and $F_2$ are compatible if and only if there are two commuting self-adjoint operators $A^+_1$ and $A^+_2$ in an extended Hilbert space $\mathcal{H}^+$ such that $F_i(\Delta)=P\chi_{\Delta}(A_i^+)P$, $i=1,2$.
\end{theorem}

Theorems \ref{CNS} and \ref{CNS1} 
are illustrated in the following diagram.

\begin{equation*}
\xymatrix{ 
A_1^+\ar@{<->}[r]&E^+_1\ar@{<->}[r]^{c}& E^+_2\ar@{<->}[r] &A_2^+\\ 
& F_1 \ar@{<->}[u]_{P}\ar@{<->}[r]^{c}& F_2\ar@{<->}[u]^{P} &
}
\end{equation*}

\noindent
where the arrow  $\xymatrix{ \ar@{<->}[r]^{c} &}$ denotes compatibility, $\xymatrix{ \ar@{<->}[r]^{P} &}$ denotes the relationship between a POVM and its dilation as expressed by the Naimark's theorem. In the case of real POVMs, the dilations $E^+_1$ and $E^+_2$ correspond to two self-adjoint operators $A^+_1$ and $A^+_2$ respectively and we use the arrow $\xymatrix{ \ar@{<->}[r] &}$ in order to represent such a correspondence.

\noindent
Therefore, we can say that each couple of compatible self-adjoint operators in an extended Hilbert space $\mathcal{H}^+$ corresponds to a couple of compatible real POVMs in $\mathcal{H}$ and vice versa.

Another possible statement of theorem \ref{CNS1} is the following.
\begin{theorem} Two real POVMs $F_1$ and $F_2$ are compatible if and only if there is an extended Hilbert space $\mathcal{H}^+$ and two commuting self-adjoint operators $A_1^+$, $A_2^+$ such that, for each bounded measurable function $f$, the operator  $f(A_i^+)$ is a self-adjoint dilation of the operator $F_i(f)=\int f(t)\,dF_i(t)$, $i=1,2$.
\end{theorem} 
\begin{proof}
Suppose $A^+_1$ and $A^+_2$ are such that $Pf(A_i^+)P=F_i(f)$ for each bounded, measurable function. Then, by setting $f=\chi_{\Delta}$ we get $PE_i^+(\Delta)P=P\chi_{\Delta}(A_i^+)P=F_i(\chi_{\Delta})=F_i(\Delta)$ which prove that the spectral measure $E_i^+$ corresponding to $A^+_i$ is a dilation of $F_i$. Moreover, by hypothesis, $[E^+_1,E^+_2]=\mathbf{0}$ and, by theorem \ref{CNS}, $F_1$ and $F_2$ are compatible.  

Now, suppose that $F_1$ and $F_2$ are compatible. By theorem \ref{CNS} there are two compatible PVMs $E^+_1$ and $E_2^+$ such that $PE^+_iP=F_i$, $i=1,2$. The self-adjoint operators $A_1^+$ and $A^+_2$ corresponding to $E^+_1$ and $E_2$ respectively commute. By theorem \ref{Pro},  $Pf(A_i^+)P=\int f(t)\,dF_i(t)=F_i(f)$, $i=1,2$.
\end{proof}
\noindent
The theorem is illustrated by the following diagram

\begin{equation*}
\xymatrix{ 
f(A^+_1)\ar[d]_{P} \ar@{<->}[r]^{c}  & f(A^+_2) \ar[d]^{P}\\ 
 F_1(f) \ar@{<-}[d]_{f}&  F_2(f)\ar@{<-}[d]_{f}\\   
F_1\ar@{<->}[r]^{c} & F_2 
}
\end{equation*}

\noindent
where, $\xymatrix{ \ar[r]_{P} &}$ denotes the projection from the extended Hilbert space $\mathcal{H}^+$ onto $\mathcal{H}$ while $\xymatrix{ \ar@{->}[r]^{f} &}$ denotes the maps $f\mapsto F(f)$. Notice that in general $F_1(f)$ and $F_2(f)$ as well as $F_1$ and $F_2$ do not commute.

\section{Compatibility and smearing}
There are well known examples of incompatible PVMs that can be smeared into two compatible POVMs. As a relevant example we can consider the position and momentum observables which are represented by two incompatible PVMs $Q$ and $P$ \cite{Davies,Holevo,Busch,Guz,Schroeck1}. By an appropriate choice of the  smearing of $Q$ and $P$ one can get two compatible POVMs $F^Q$ and $F^P$ (see example \ref{QP}). Another relevant example was recently provided in Ref. \cite{Busch1}. Here it is shown that any couple of observables $F_1$, $F_2$ in a general probabilistic model can always be smeared in such a way to get two compatible observables. That is relevant since it provides a transition from incompatibility to compatibility for any couple of incompatible observables. Now, we use the same kind of smearing in the quantum mechanical context where observables are represented by POVMs. In particular, given two POVMs $F_1$, $F_2$, the smearings 
\begin{align*}
\widetilde{F}_1(\Delta_1)&=\int\mu^{(1)}_{\Delta_1}(x)\,dF_1(x)=\int\big[\lambda\chi_{\Delta_1}(x)+(1-\lambda)\nu^{(1)}(\Delta_1)\big]\,dF_1(x)\\
 \widetilde{F}_2(\Delta_2)&=\int\mu^{(2)}_{\Delta_2}(x)\,dF_2(x)=\int\big[(1-\lambda)\chi_{\Delta_2}(x)+\lambda\nu^{(2)}(\Delta_2)\big]\,dF_2(x)
\end{align*}
are compatible. Indeed,
$$\widetilde{F}(\Delta_1\times\Delta_2)=\lambda\nu^{(2)}(\Delta_2)F_1(\Delta_1)+(1-\lambda)\nu^{(1)}(\Delta_1)F_2(\Delta_2).$$
\noindent
is a joint POVM.

\noindent
That raises the problem of characterizing those smearings which convert two incompatible POVMs into two compatible ones. The aim of the present section is to give conditions for the compatibility of the smearing of two incompatible real PVMs (or self-adjoint operators). That is equivalent to give conditions for the compatibility of two commutative POVMs. In particular, we establish a connection between the compatibility of the smearings of two self-adjoint operators $A_1$, $A_2$ and the existence of two compatible self-adjoint dilations of $A_1$ and $A_2$.

\begin{proposition}
If $A_1$ and $A_2$ are compatible, all the smearings $F_1$, $F_2$ of $A_1$ and $A_2$ are compatible.
\end{proposition}
\begin{proof}
Let $F_1$ and $F_2$ be smearings of $A_1$ and $A_2$ respectively. Then there are two Markov kernels $\mu^{(1)}$ and $\mu^{(2)}$ such that $F_i(\Delta)=\mu^{(i)}_{\Delta}(A_i)$, $i=1,2$. Since $[A_1,A_2]=\mathbf{0}$, we have $[F_1,F_2]=\mathbf{0}$. Therefore, by theorem \ref{commuting}, $F_1$ and $F_2$ are compatible.
\end{proof}

\begin{theorem}\label{bijective}
Let $F_1$ and $F_2$ be compatible smearings of $A_1$ and $A_2$ respectively. Then, they are compatible smearings of $f_1(A_1)$ and $f_2(A_2)$ whenever $f_1:\sigma(A_1)\to\mathbb{R}$ and $f_2:\sigma(A_2)\to\mathbb{R}$ are almost everywhere bijective. 
\end{theorem}

\begin{proof}
Since $F_1$, $F_2$ are smearings of $A_1$ and $A_2$ respectively, there are two Markov kernels $\mu^{(1)}$, $\mu^{(2)}$ such that $F_1(\Delta)=\mu^{(1)}_{\Delta}(A_1)$, $F_2(\Delta)=\mu^{(2)}_{\Delta}(A_2)$. Now, we can define the Markov kernels $\omega_{\Delta}^{(1)}(\lambda):=[\mu^{(1)}_{\Delta}\circ f_1^{-1}](\lambda)=\mu^{(1)}_{\Delta}( f_1^{-1}(\lambda))$, $\omega^{(2)}_{\Delta}(\lambda):=[\mu^{(2)}_{\Delta}\circ f_2^{-1}](\lambda)=\mu^{(2)}_{\Delta}( f_2^{-1}(\lambda))$. We have 

$$\omega^{(1)}_{\Delta}(f_1(A_1))=\mu^{(1)}_{\Delta}( f_1^{-1}(f_1(A_1)))=F_1(\Delta)$$
$$\omega^{(2)}_{\Delta}(f_2(A_2))=\mu^{(2)}_{\Delta}( f_2^{-1}f_2((A_2)))=F_2(\Delta)$$

\noindent
which proves the thesis.
\end{proof}

\begin{theorem}
If two bounded self-adjoint operators $A_1$ and $A_2$ have compatible bounded self-adjoint dilations $A^+_1$ and $A_2^+$ such that $F_i(\Delta):=P\chi_{\Delta} (A^+_i)P$ is a commutative POVM and $\mathcal{A}^W(F_i)=\mathcal{A}^W(A_i)$, $i=1,2$, then $F_1$ and $F_2$ are compatible smearing of $A_1$ and $A_2$ respectively.
\end{theorem}
\begin{proof}
Since $A^+_1$ and $A^+_2$ are compatible, the corresponding spectral measures $E^+_1$ and $E_2^+$ commute. Since $F_i(\Delta):=P\chi_{\Delta} (A^+_i)P=PE^+_i(\Delta)P$, theorem \ref{CNS} assures that $F_1$ and $F_2$ are compatible. By theorem \ref{Markov}, $F_1$ and $F_2$ are smearings of their  sharp versions $B_1$ and $B_2$ respectively. Since $B_i$ and $A_i$, $i=1,2$,  generate the same von Neumann algebra, there are one-to-one functions $f_i$ such that $f_i(A_i)=B_i$, $i=1,2$. By theorem \ref{bijective}, $F_i$ is a smearing of $A_i$ as well.
\end{proof}

\vskip.5cm
\noindent
The theorem is illustrated by the following diagram.
\vskip.5cm
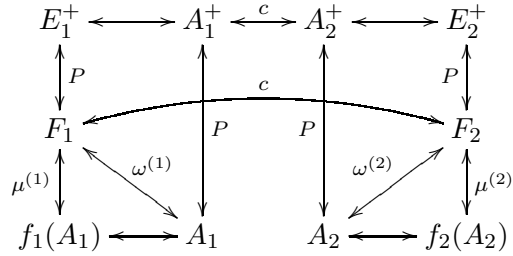
\begin{equation*}
\xymatrix{ 
E_1^+\ar@{<->}[d]^{P}\ar@{<->}[r] & A^+_1\ar@{<->}[dd]^{P}\ar@{<->}[r]^{c} &  A^+_2 \ar@{<->}[dd]_{P}\ar@{<->}[r] & E_2^+\ar@{<->}[d]_{P} \\ 
F_1\ar@{<->}@/^1pc/[rrr]^{c}\ar@{<->}[rd]^{\omega^{(1)}}&  & &F_2\\
f_1( A_1)\ar@{<->}[u]^{\mu^{(1)}}\ar@{<->}[r]^{}& A_1 &  A_2\ar@{<->}[ur]^{\omega^{(2)}}\ar@{<->}[r]^{} & f_2(A_2)\ar@{<->}[u]_{\mu^{(2)}} 
}
\end{equation*}

\noindent
where $\xymatrix{ \ar@{<->}[r]^{P} &}$ denotes the relationship between a self-adjoint operator and its dilation as well as the relationship between a POVM and its Naimark's dilation while $\xymatrix{ \ar@{<->}[r]_{} &}$ denotes the equivalence of two self-adjoint operators as expressed in definition \ref{one-to-one} as well as the equivalence of a self-adjoint operator and the corresponding spectral measure (up to bijections). The functions $f_i$ in the diagram are one-to-one and both $A_i$ and $f_i(A_i)$ are sharp versions of $F_i$, $i=1,2$. Moreover, $\omega^{(i)}=\mu^{(i)}\circ f_i^{-1}$ and the arrow $\xymatrix{ \ar@{<->}[r]^{\mu} &}$ denotes equivalence in the sense specified in remark \ref{information}.

\vskip.5cm

\noindent
As a corollary of  the previous results we have that each couple of self-adjoint operators admits (up to bijections) a couple of compatible self-adjoint dilations.
\begin{corollary}
Let $A_1$ and $A_2$ be discrete self-adjoint operators. Then, there are two one-to-one functions $h_1$, $h_2$ such that $h_1(A_1)$ and $h_2(A_2)$ admit compatible self-adjoint dilations $A_1^+$ and $A_2^+$ respectively.
\end{corollary}

\begin{proof}
Let $E_1$ and $E_2$ be the spectral measures corresponding to $A_1$ and $A_2$ respectively. We can define the following compatible smearing of $E_1$ and $E_2$. 
\begin{align*}
F_1(\Delta_1)&=\int\mu^{(1)}_{\Delta_1}(x)\,dE_1(x)=\int\big[\lambda\chi_{\Delta_1}(x)+(1-\lambda)\nu^{(1)}(\Delta_1)\big]\,dE_1(x)\\
F_2(\Delta_2)&=\int\mu^{(2)}_{\Delta_2}(x)\,dE_2(x)=\int\big[(1-\lambda)\chi_{\Delta_2}(x)+\lambda\nu^{(2)}(\Delta_2)\big]\,dE_2(x)
\end{align*}
They are compatible since 
$$F(\Delta_1\times\Delta_2)=\lambda\nu^{(2)}(\Delta_2)E_1(\Delta_1)+(1-\lambda)\nu^{(1)}(\Delta_1)E_2(\Delta_2).$$
\noindent
is a joint POVM. By theorem \ref{CNS} there are two compatible Naimark's extensions $E^+_1$, $E_2^+$ corresponding to two self-adjoint operators $B_1^+$ and $B_2^+$. By theorem \ref{Be2}, there are one-to-one functions $f_1$, $f_2$ and $h_1$, $h_2$ such that $A_1^+:=f_1(B^+_1)$ and $A_2^+:=f_2^+(B_2^+)$ are dilations of $h_1(A_1)$ and $h_2(A_2)$ respectively\footnote{We have used the fact that $F(\Delta)=\mu_{\Delta}(A)$ is discrete if $A$ is discrete.}. 

\end{proof}

\noindent
The corollary is illustrated by the following diagram.
\vskip.5cm

\begin{equation*}
\xymatrix{ 
 & A^+_1\ar@{<->}[d]_{P}\ar@{<->}[r]^{c} &  A^+_2 \ar@{<->}[d]^{P} & \\ 
 A_1\ar@{<->}[r]^{}& h_1(A_1) & h_2( A_2)\ar@{<->}[r]^{} & A_2 
}
\end{equation*}
\vskip.5cm
\noindent

\vskip.5cm
\noindent
Now, we give a necessary and sufficient condition for two commutative POVMs $F_1$ and $F_2$ to be compatible which is based on the existence of two commuting self-adjoint dilations $A_1^+$ and $A^+_2$ of the sharp versions $A_1$ and $A_2$ respectively. 

\begin{theorem}\label{self}
Let $F_1$ and $F_2$ be two commutative POVMs such that the operators in their ranges are discrete. They are compatible if and only if the corresponding sharp versions $A_1$ and $A_2$ can be dilated to two compatible self-adjoint operators $A_1^+$, $A_2^+$ such that $P\chi_{\Delta}(A_i^+)P=F_i(f_i^{-1}(\Delta))$, $i=1,2$, with $f_i$ one-to-one.  
\end{theorem}
\begin{proof}
Suppose $F_1$ and $F_2$ to be compatible. Then, by theorem \ref{CNS}, there are two PVMs $E^+_1$ and $E_2^+$ such that $[E^+_1,E^+_2]=\mathbf{0}$ and $F_i(\Delta)=PE_i^+(\Delta)P=P\chi_{\Delta}(B_i^+)P$ where,  $B_i^+$ is the self-adjoint operator corresponding to $E_i^+$. By theorems \ref{Be2} and \ref{bijections}, we have 
\begin{align*}
PA^+_iP:=Pf_i(B^+_i)P=P\int f_i(\lambda)\,dE^+_i(\lambda)P=\int f_i(t)\,dF_i(t)=A_i\quad i=1,2
\end{align*} 
where, $f_i$ is one-to-one, $A_i^+=f_i(B_i^+)$ and $A_i$ is the sharp version of $F_i$, $i=1,2$. Therefore, $A^+_1$ and $A^+_2$ are commuting dilations of $A_1$ and $A_2$ respectively. Moreover, 
\begin{align*}
P\chi_{\Delta}(A_i^+)P&=P\int (\chi_{\Delta}\circ f_i)(\lambda)\,dE^+_i({\lambda})\,\,P\\
&=P\int \chi_{f_i^{-1}(\Delta)}(\lambda)\,dE^+_i({\lambda})\,\,P=PE^+(f_i^{-1}(\Delta))P=F_i(f_i^{-1}(\Delta)).
\end{align*}

\noindent
Conversely, suppose that $A_1^+$ and $A^+_2$ are compatible dilations of the sharp versions $A_1$ and $A_2$ respectively and that  $F_i(f_i^{-1}(\Delta))=P\chi_\Delta(A_i^+)P$, $i=1,2$. Then, $F_i(\Delta)=P\chi_{f_i(\Delta)}(A_i^+)P=PE^{A_i^+}(f_i(\Delta))P=PE_i^+(\Delta)P$, $i=1,2$, where, $E^{A_i^+}=E_i^+\circ f_i^{-1}$ is the spectral measure corresponding to $A^+_i$. Hence, $E_1^+$ and $E_2^+$ are compatible Naimark's extensions of $F_1$ and $F_2$ respectively and theorem \ref{CNS} ends the proof.
\end{proof}

\vskip.5cm
\noindent
The following diagram illustrates theorem \ref{self}.

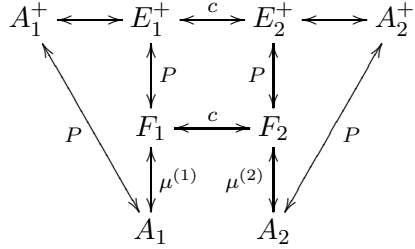
\begin{equation*}
\xymatrix{ 
A^+_1\ar@{<->}[ddr]_{P} \ar@{<->}[r]^{} & E^+_1 \ar@{<->}[r]^{c} &E^+_2\ar@{<->}[r]^{} &A^+_2 \ar@{<->}[ddl]^{P}\\ 
 & F_1\ar@{<->}[r]^{c}\ar@{<->}[u]_{P}  &  F_2\ar@{<->}[u]^{P} & \\
& A_1 \ar@{<->}[u]_{\mu^{(1)}} & A_2 \ar@{<->}[u]^{\mu^{(2)}}  &
}
\end{equation*}

\vskip1cm

\noindent
Next we illustrate theorem \ref{self} by means of a relevant physical example.

\noindent
In the following $\ast$ denotes convolution, i.e., $(h\ast g)(x)=\int_{\mathbb{R}} h(y)g(x-y)d y$ while $\hat{g}$ denotes the Fourier Transform of $g$. 

\begin{example}\label{QP}
As a relevant physical example, we consider the position and momentum observables, $Q=\int q\,dQ(q)$ and $P=\int p\,dP(p)$ on the space $\mathcal{H}=L^2(\mathbb{R})$. We recall that $(Q\psi)(q)=q\,\psi(q)$ while $(P\psi)(q)=-i\frac{\partial \psi}{\partial q}(q)$.

It is possible to introduce compatible smearings $F^Q$ and $F^P$ of $Q$ and $P$ respectively and then to build a joint POVM for $F^Q$ and $F^P$. We do the converse, i.e., we  start from a POVM $F$ on the phase space $\Gamma=\mathbb{R}\times\mathbb{R}$ and show that its marginals are smearings of $Q$ and $P$ respectively.

\noindent
Let us consider the joint position-momentum POVM \cite{Ali,Busch,Davies,Guz,Holevo,Prugovecki,Schroeck1,Stulpe} 
\begin{equation*} \label{phase}
F(\Delta\times\Delta')=\int_{\Delta\times\Delta'}U_{q,p}\,\eta\,U^*_{q,p}\,d q\, d p=\int_{\Delta\times\Delta'}P_{q,p}\,d q\, d p
\end{equation*}
where, $U_{q,p}=e^{-iqP}e^{ipQ}$, $\eta:=P_{g}$ is the projector on the subspace generated by $g\in L^2(\mathbb{R})$, $\|g\|_2=1$ and $P_{q,p}=U_{q,p}\,\eta\,U^*_{q,p}$.
The marginals
\begin{align}
\label{approximate}
F^Q_{g}(\Delta)&:=F(\Delta\times\mathbb{R})=\int_{-\infty}^{\infty}({\bf 1}_{\Delta}\ast\vert g\vert^2)(q)\,dQ(q),\quad\Delta\in\mathcal{B}(\mathbb{R}),\\
F^P_{g}(\Delta)&:=F(\mathbb{R}\times\Delta)=\int_{-\infty}^{\infty}({\bf 1}_{\Delta}\ast\vert \hat{g}\vert^2)(p)\,dP(p),\quad\Delta\in\mathcal{B}(\mathbb{R})
\end{align}
are the unsharp position and momentum observables respectively (\cite{Davies,Schroeck1,Busch,B4,B9,B5}).  Notice that the maps $\mu_{\Delta}(q):=({\bf 1}_{\Delta}\ast \vert g\vert^2)(q)$ and $\hat{\mu}_{\Delta}(p):=({\bf 1}_{\Delta}\ast \vert \hat{g}\vert^2)(p)$ define two Markov kernels (\cite{B4,B5,B9}). 
Now, we can define the isometry 
\begin{align*}
W^\eta:\mathcal{H}&\to L^2(\Gamma,\mu)\\
\psi&\mapsto\langle U_{q,p}\,g,\psi\rangle
\end{align*}
\noindent
where, $\mu$ is the Lebesgue measure on $\Gamma=\R\times\R$. The map $W^\eta$ embeds $\mathcal{H}$ as a subspace of $L^2(\Gamma, \mu)$. The projection operator $\widetilde{P}^\eta$ from $L^2(\Gamma, \mu)$ to $W^\eta(\mathcal{H})$ is defined as follows
$$(\widetilde{P}^\eta f)(q,p)=\int_{\Gamma}\langle U_{q,p}\,g,U_{q',p'}\,g\rangle f(q',p')\,dq'dp'.$$

Next, we prove the existence of two commuting Naimark's dilations for $F^Q_g$ and $F^P_g$. It is sufficient to consider the following two PVMs
\begin{align*}
(\widetilde{E}^+_Q(\Delta)f)(q,p)=\chi_{\Delta}(q)f(q,p),\quad f\in L^2(\Gamma,\mu)\\
(\widetilde{E}_P^+(\Delta)f)(q,p)=\chi_{\Delta}(p)f(q,p),\quad f\in L^2(\Gamma,\mu)
\end{align*}   
They commute since they are multiplications by characteristic functions. Moreover, for any $f\in W^\eta(\mathcal{H})$,
\begin{align*}
(\widetilde{P}^\eta \widetilde{E}_Q^+(\Delta)f)(q,p)&=\int_{\Gamma}\langle U_{q,p}g,U_{q',p'}g\rangle \chi_{\Delta}(q')f(q',p')\,dq'dp'\\
&=\int_{\Delta\times\R}\langle U_{q,p}g,U_{q',p'}g\rangle f(q',p')\,dq'dp'\\
&=\int_{\Delta\times\R}\langle U_{q,p}g,U_{q',p'}g\rangle \langle U_{q',p'}\,g,\psi\rangle\,dq'dp'\\
&=\int_{\Delta\times\R}\langle U_{q,p}g,U_{q',p'}\,g\rangle\langle g,\,U^*_{q',p'}\psi\rangle\,d q'\, d p'\\
&=W^\eta\int_{\Delta\times\R}U_{q',p'}\,g\,\langle g,\,U^*_{q',p'}\psi\rangle\,d q'\, d p'\\
&=W^\eta\int_{\Delta\times\R}U_{q',p'}\,\eta\,U^*_{q',p'}\psi\,d q'\, d p'\\
&=[W^\eta F^Q_{g}(\Delta)(W^\eta)^{-1}f](q,p).
\end{align*} 

\noindent
which proves that $ \widetilde{E}_Q^+$ is a Naimark's dilation of $W^{\eta}\,F^Q_g\,(W^{\eta})^{-1}$. An analogous argument holds for $\widetilde{E}_P^+$ and $W^{\eta}\,F^Q_g\,(W^{\eta})^{-1}$.

Now, if we specialize ourselves to the case $g=\frac{1}{l\,\sqrt{2\,\pi}}\,e^{(-\frac{x^2}{2\,l^2})}$, $l\in\mathbb{R}-\{0\}$, we get, \cite{B4}
\begin{align*}
\widetilde{P}^{\eta}\Big(\int t\,d\widetilde{E}_Q^+(t)\Big)\widetilde{P}^\eta&=W^{\eta}\int t\,dF^Q_g(t)\,(W^{\eta})^{-1}=W^{\eta}\,Q\,(W^{\eta})^{-1}\\
\widetilde{P}^{\eta} \Big(\int t\,d\widetilde{E}_P^+(t)\Big)\widetilde{P}^\eta&=W^{\eta}\int t\,dF^P_g(t)\,(W^{\eta})^{-1}=W^{\eta}\,P\,(W^{\eta})^{-1}.
\end{align*}
\noindent
Therefore, the compatible operators $Q^+:=\int t\,d\widetilde{E}_Q^+(t)$ and $P^+:=\int t\,d\widetilde{E}_P^+(t)$ are dilations of $W^{\eta}\,Q\,(W^{\eta})^{-1}$ and $W^{\eta}\,P\,(W^{\eta})^{-1}$ respectively.  All that is summarized (up to isometry) in the following commuting diagram.

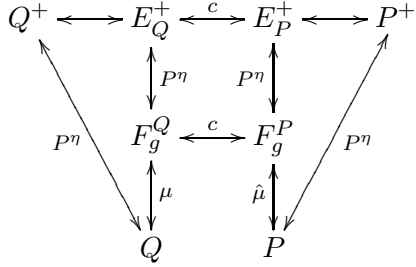
\begin{equation*}
\xymatrix{ 
Q^+\ar@{<->}[ddr]_{P^{\eta}} \ar@{<->}[r]^{} & E_Q^+ \ar@{<->}[r]^{c} &E_P^+\ar@{<->}[r]^{} &P^+ \ar@{<->}[ddl]^{P^{\eta}}\\ 
 & F^Q_g\ar@{<->}[u]_{P^{\eta}}\ar@{<->}[r]^{c}  &  F^P_g\ar@{<->}[u]^{P^\eta} & \\
& Q \ar@{<->}[u]_{\mu} & P \ar@{<->}[u]^{\hat{\mu}}  &
}
\end{equation*}

\end{example}

\section{Compatibility between effects}

In the present section we recall the definition of compatibility between two effects and show that two effects are compatible if and only if they can be dilated to two commuting projections. Then, we prove a sufficient condition of the compatibility.

\begin{definition}
Two effects $A_1$, $A_2$ are compatible if the POVMs $F_1=\{A_1, \mathbf{1}-A_1\}$ and $F_2=\{A_2, \mathbf{1}-A_2\}$ are compatible. 
\end{definition}
\noindent
Notice that the definition of compatibility between effects refers to the corresponding dicotomic POVMs and is therefore different from the joint measurability between self-adjoint operators in definition \ref{PVM}.
\begin{theorem}\label{effects}
Two effects $A_1$, $A_2$ are compatible if and only if there are two commutative  projections $E_1$, $E_2$ in an extended Hilbert space $\mathcal{H}^+$ such that $PE_1^+P=A_1$ and $PE_2^+P=A_2$.  
\end{theorem}
\begin{proof}
By theorem \ref{CNS}, $\{A_1, \mathbf{1}-A_1\}$ and $\{A_2, \mathbf{1}-A_2\}$ are compatible if and only if there are two compatible Naimark extensions $\{E^+_1,\mathbf{1}-E_1^+\}$, $E^+_2,\mathbf{1}-E_2^+\}$. In particular, $A_1=PE^+_1P$, $A_2=PE^+_2P$ and $[E_1^+,E_2^+]=\mathbf{0}$. 
\end{proof}

\subsection{A condition for the compatibility of two effects}
Let $A_1$ and $A_2$ be two effects in $\mathcal{H}$. We can dilate them to two projections in an extended Hilbert space $\mathcal{H}^+$ by means of the following procedure (see \cite{Riesz}, page 461).

\noindent
Let $\mathcal{H}^+=\mathcal{H}\times\mathcal{H}$. The Hilbert space $\mathcal{H}$ can be embedded in $\mathcal{H}^+$ by identifying $\psi$ with $\begin{pmatrix}
\psi\\
0
\end{pmatrix}$. Next, we write the elements of $\mathcal{H}^+$ as column vectors 
$\begin{pmatrix}
\psi_1\\
\psi_2
\end{pmatrix}$
and operators on $\mathcal{H}^+$ as matrices $\begin{pmatrix}
A_{1,1} & A_{1,2}\\
A_{2,1} & A_{2,2}
\end{pmatrix}$
where, $A_{i,j}$ is a bounded self-adjoint operator on $\mathcal{H}$.

Now, for each effect $A_i$ we define the operator 
\begin{equation}\label{extension}E_i^+=\begin{pmatrix}
A_i & \sqrt{A_i(\mathbf{1}-A_i)}\\
\sqrt{A_i(\mathbf{1}-A_i)} & \mathbf{1}-A_i
\end{pmatrix}
\end{equation}
\noindent
Notice that 
\begin{equation*}P\begin{pmatrix}
A_i & \sqrt{A_i(\mathbf{1}-A_i)}\\
\sqrt{A_i(\mathbf{1}-A_i)} & \mathbf{1}-A_i
\end{pmatrix}P=
A_i
\end{equation*}
\noindent
and 
\begin{equation*}\begin{pmatrix}
A_i & \sqrt{A_i(\mathbf{1}-A_i)}\\
\sqrt{A_i(\mathbf{1}-A_i)} & \mathbf{1}-A_i
\end{pmatrix}^2=
\begin{pmatrix}
A_i & \sqrt{A_i(\mathbf{1}-A_i)}\\
\sqrt{A_i(\mathbf{1}-A_i)} & \mathbf{1}-A_i
\end{pmatrix}
\end{equation*}
\noindent
where, $P$ is the projection from $\mathcal{H}^+$ onto $\mathcal{H}$, i.e., $P\begin{pmatrix}
\psi_1\\
\psi_2
\end{pmatrix}=\begin{pmatrix}
\psi_1\\
0
\end{pmatrix}$.
\noindent
By means of the Naimark dilation we just introduced, we can state the following condition for the compatibility of two effects. In the following $B_1=\sqrt{A_1(\mathbf{1}-A_1)}$ and $B_2=\sqrt{A_2(\mathbf{1}-A_2)}$.
\begin{proposition}\label{condition}
 Two effects $A_1$, $A_2$ such that $[A_1,A_2]+[B_1,B_2]=\mathbf{0}$ and $\{A_1,B_2\}-\{B_1,A_2\}=B_2-B_1$ are compatible.
\end{proposition}
\begin{proof}
By theorem \ref{effects}, $A_1$ and $A_2$ are compatible if and only if they can be extended to two commuting projections $E_1^+$ and $E^+_2$ respectively. Let us consider the extension $E_1^+$ and $E^+_2$ in equation (\ref{extension}) and prove that they commute.

We have, 
\begin{equation*}
E^+_1E^+_2=\begin{pmatrix}
A_1A_2+B_1B_2 & A_1B_2+B_1(\mathbf{1}-A_2)\\
B_1A_2+B_2-A_1B_2 & B_1B_2+ \mathbf{1}-A_2-A_1+A_1A_2
\end{pmatrix}
\end{equation*}
\noindent
and
\begin{equation*}
E^+_2E^+_1=\begin{pmatrix}
A_2A_1+B_2B_1 & A_2B_1+B_2(\mathbf{1}-A_1)\\
B_2A_1+B_1-A_2B_1 & B_2B_1+ \mathbf{1}-A_2-A_1+A_2A_1
\end{pmatrix}
\end{equation*}
\noindent
The two matrices coincide if and only if
\begin{align*}
A_1A_2+B_1B_2&=A_2A_1+B_2B_1\\
A_1B_2+B_1(\mathbf{1}-A_2)&=A_2B_1+B_2(\mathbf{I}-A_1)
\end{align*} 
\noindent
which proves the thesis.
\end{proof}

\end{document}